 \documentclass[proceedings]{stacs}
  \stacsheading{2010}{453-464}{Nancy, France}
\firstpageno{453}
\usepackage{amsmath}
\usepackage{amssymb}
\usepackage{amsfonts}
\usepackage{amsthm}
\newcommand\vhdef2

\newcommand{\class}[1]{{\ifnum\vhdef=2\mathbf{#1}\else\mathrm{#1}\fi}}
\newcommand{\co}{
\ifnum\vhdef=2\mathbf{co\,}\else\mathrm{co\hspace{2pt}}\fi
\ifnum\vhdef=2\textbf{-}\else\textrm{-}\fi
}
\newcommand{\NP}{\class{NP}}
\newcommand{\coNP}{\class{{\co}NP}}

\newcommand{\substd}{16d\log^2_*n}
\newcommand{\substf}{17d\log_*^2n}
\newcommand{\substk}{12 d' \ln (16 d\log_*^2n)} 
\newcommand{\substl}{(3.03\cdot 10^4)\cdot f\cdot\ln (16kd\log_*^2n)} 
\newcommand{\errfhigh}{\frac1{0.99f}}
\newcommand{\errhighsubst}{\frac1{16d\log_*^2n}} 

\newcommand{\inverrlowsubst}{18d\log_*^2n}
\newcommand{\errlowsubst}{\frac1{\inverrlowsubst}}
\newcommand{\testerrflowproved}{\frac1{1.01f}}
\newcommand{\testerrflowused}{\frac1{1.011f}}
\newcommand{\testerrrej}{e^{-\frac{l}{3.03\cdot 10^4\cdot f}}}
\newcommand{\testerracc}{e^{-\frac{l}{2\cdot 10^4\cdot f}}}
\newcommand{\certerrrej}{e^{-\frac{k}{12d'}} + k\cdot \testerrrej}
\newcommand{\certerracc}{e^{-\frac{k}{8d'}}  + k\cdot \testerracc}
\newcommand{\certerrrejsubst}{\frac1{8 d\log_*^2n}} 
\newcommand{\certerraccsubst}{\frac1{8 d\log_*^2n}} 
\newcommand{\ournote}[1]{} 
\newcommand{\ourcomment}[1]{} 
\begin{document}
\title[On optimal heuristic randomized semidecision procedures]
{On optimal heuristic randomized semidecision procedures,
with application to proof complexity}
\author{E. A. Hirsch}{Edward A. Hirsch}
\author{D. Itsykson}{Dmitry Itsykson}
\address{Steklov Institute of Mathematics at St. Petersburg,
   \newline 27 Fontanka, St.Petersburg, 191023, Russia}  
\urladdr{http://logic.pdmi.ras.ru/\~{}hirsch}  
\urladdr{http://logic.pdmi.ras.ru/\~{}dmitrits}

\thanks{
Partially supported by
grants RFBR 08-01-00640 and 09-01-12066,
and the president of Russia grant ``Leading Scientific Schools'' NSh-4392.2008.1,
by Federal Target Programme ``Scientific and
scientific-pedagogical personnel of the innovative Russia'' 2009-2013
(contract N $\Pi$265 from 23.07.2009). The second author is also supported by
Russian Science Support Foundation.
}
%



\begin{abstract}
The existence of a ($p$-)optimal propositional proof system
is a major open question in (proof) complexity;
many people conjecture that such systems do not exist.
Kraj\'{\i}\v{c}ek and Pudl\'{a}k \cite{KP}
show that this question is equivalent
to the existence of an algorithm that is
optimal\footnote{Recent papers \cite{Monroe}
call such algorithms \emph{$p$-optimal} while
traditionally Levin's algorithm was called \emph{optimal}.
We follow the older tradition. Also there is some
mess in terminology here, thus please see formal definitions
in Sect.~\ref{sec:prelim} below.}
on all propositional tautologies.
Monroe \cite{Monroe} recently gave a conjecture 
implying that such algorithm does not exist.

We show that in the presence of errors such optimal 
algorithms \emph{do} exist. The concept is motivated
by the notion of heuristic algorithms.
Namely, we allow the algorithm to claim
a small number of false ``theorems''
(according to any polynomial-time samplable distribution on non-tautologies)
and err with bounded probability on other inputs.

Our result can also be viewed as the existence of an optimal
proof system in a class of proof systems obtained by
generalizing automatizable proof systems.

\bigskip
\end{abstract}

\maketitle

\section{Introduction}\label{sec:intro}
Given a specific problem, does there exist the ``fastest''
algorithm for it? Does there exist a proof system
possessing the ``shortest'' proofs of the positive
solutions to the problem? 
Although the first result in this direction 
was obtained by Levin \cite{Levin-optimal} in 1970s,
these important questions are still open for most interesting languages,
for example, the language of propositional tautologies.

\paragraph{Classical version of the problem.}
According to Cook and Reckhow \cite{CR79}, a proof system
is a polynomial-time mapping of all strings (``proofs'') onto ``theorems''
(elements of certain language $L$; if $L$ is the language of
all propositional tautologies, the system is called a \emph{propositional}
proof system).
The existence of a \emph{polynomially bounded}
propositional proof system (that is, a system that has a
polynomial-size proof for every tautology) is equivalent to $\NP=\coNP$.
In the context of polynomial boundedness a proof system can be 
equivalently viewed as a function that given a formula and a ``proof'',
verifies in polynomial time that a formula is a tautology:
it must accept at least one ``proof'' for each tautology (\emph{completeness})
and reject all proofs for non-tautologies (\emph{soundness}).

One proof system $\Pi_w$ is \emph{simulated} by another one $\Pi_s$ if the shortest proofs
for every tautology in $\Pi_s$ are at most polynomially
longer than the shortest proofs in $\Pi_w$.
The notion of \emph{$p$-simulation} is similar, but requires also a polynomial-time
computable function for translating the proofs from $\Pi_w$ to $\Pi_s$.
A \emph{($p$-)optimal} propositional proof system is one that ($p$-)simulates all other
propositional proof systems. 

The existence of an optimal (or $p$-optimal) propositional
proof system is a major open question.
If one would exist, it would allow to reduce
the $\NP$ vs $\coNP$ question to proving proof size bounds
for just one proof system.
It would also imply the existence of a complete disjoint $\mathbf{NP}$ pair
\cite{Razborov,Pudlak}. 
Kraj\'{\i}\v{c}ek and Pudl\'{a}k \cite{KP}
show that the existence of a $p$-optimal system is equivalent
to the existence of an algorithm 
that is optimal on all propositional tautologies,
namely, it always solves the problem correctly
and it takes for it at most polynomially longer
to stop on every tautology than for any other correct algorithm
\emph{on the same tautology}.
Monroe \cite{Monroe} recently gave a conjecture 
implying that such algorithm does not exist.
Note that Levin \cite{Levin-optimal} showed
that an optimal algorithm does exist for finding witnesses
to non-tautologies; however, (1) its behaviour
on tautologies is not restricted; (2) after translating
to the decision problem by self-reducibility
the running time in the optimality condition is compared
to the running time for \emph{all shorter formulas as well}.

An \emph{automatizable} proof system is one that has
an automatization procedure that given a tautology,
outputs its proof of length polynomially bounded by
the length of the shortest proof
in time bounded by a polynomial in the output length.
The automatizability of a proof system $\Pi$ implies
polynomial separability of its canonical $\mathbf{NP}$ pair \cite{Pudlak},
and the latter implies the automatizability
of a system that $p$-simulates $\Pi$.
This, however, does not imply the existence
of ($p$-)optimal propositional proof systems
in the class of automatizable proof systems.
To the best of our knowledge, no such
system is known to the date.

\paragraph{Proving propositional tautologies heuristically.}
An obvious obstacle to constructing an optimal proof system by
enumeration is that no efficient procedure is known for enumerating
the set of all complete and sound proof systems.
Recently a number of papers overcome similar obstacles in 
other settings by considering either
computations with non-uniform advice
(see \cite{Fortnow-survey} for survey)
or
\emph{heuristic} algorithms
\cite{FS-old,Pervyshev-heuristic,Itsykson}.
In particular, optimal propositional proof systems with advice
do exist \cite{CK}. 
We try to follow the approach of heuristic computations to
obtain a ``heuristic'' proof system.
While our work is motivated by propositional proof complexity,
i.e., proof systems for the set of propositional tautologies,
our results apply to proof systems for any recursively
enumerable language.

We introduce a notion of a \emph{randomized heuristic automatizer}
(a randomized semidecision procedure that may have false positives) 
and a corresponding notion of a \emph{simulation}.
Its particular case, a deterministic automatizer
(making no errors) for language $L$,
along with deterministic simulations,
can be viewed in two ways:
\begin{itemize}
\item as an automatizable proof system for $L$
(note that such proof system can be identified
with its automatization procedure;
however, it may not be the case for randomized algorithms,
whose running time may depend on the random coins),
where simulations are $p$-simulations of proof systems;
\item as an algorithm for $L$,
where simulations are simulations of algorithms for $L$ in the sense of \cite{KP}.
\end{itemize}

Given $x\in L$, an automatizer must return $1$ and stop.
The question (handled by simulations) is how fast it does the job.
For $x\notin L$, the running time does \emph{not} matter.
Given $x\notin L$, a deterministic automatizer simply
must \emph{not} return $1$. A randomized heuristic automatizer
may erroneously return $1$; however, for ``most'' inputs it may do it
only with bounded probability (``good'' inputs).
The precise notion of ``most''
inputs is: given an integer parameter $d$ and a sampler
for $\overline{L}$, ``bad'' inputs must have probability
less than $1/d$ according to the sampler.
The parameter $d$ is handled by simulations in the way
such that no automatizer can stop in time
polynomial in $d$ and the length of input
unless an optimal automatizer can do that.

In Sect.~\ref{sec:prelim} we give precise definitions.
In Sect.~\ref{sec:optimal} we construct an optimal randomized heuristic automatizer.
In Sect.~\ref{sec:systems} we give a notion of heuristic
probabilistic proof system and discuss the relation of automatizers
to such proof systems.


\section{Preliminaries}\label{sec:prelim}

\subsection{Distributional proving problems}
In this paper we consider algorithms
and proof systems \emph{that allow small errors},
i.e., claim a small amount of wrong theorems.
Formally, we have a probability distribution
concentrated on non-theorems
and require that the probability of sampling a non-theorem
accepted by an algorithm or validated by the system is small.
\begin{definition}
We call a pair $(D,L)$ a \emph{distributional proving problem}
if $D$ is a collection of probability distributions
$D_n$ concentrated on $\overline{L}\cap\{0,1\}^n$.
\end{definition}

In what follows we write $\Pr_{x\gets D_n}$ to denote
the probability taken over $x$ from such distribution,
while $\Pr_A$ denotes the probability taken
over internal random coins used by algorithm $A$.

\subsection{Automatizers}

\begin{definition}\label{def:aut}
A \emph{$(\lambda,\epsilon)$-correct automatizer} for 
distributional proving problem $(D,L)$ is a randomized algorithm $A$
with two parameters $x\in\{0,1\}^*$ and $d\in\mathbb{N}$
that satisfies the following conditions:
\begin{enumerate}
\item\label{def:aut:diverge} $A$ either outputs $1$ 
      (denoted $A(\ldots)=1$)
      or does not halt at all (denoted $A(\ldots)=\infty$);
\item\label{def:aut:comp} For every $x\in L$ and $d\in\mathbb{N}$, 
      $A(x,d)=1$. 
\item\label{def:aut:corr} For every $n,d\in\mathbb{N}$,
      $$\Pr_{r\gets D_n}\left\{\Pr_A\{A(r,d)=1\}>\epsilon\right\}<\frac1{\lambda d}.$$
\end{enumerate}
Here $\lambda>0$ is a constant and $\epsilon>0$ may depend on the first input ($x$) length.
An \emph{automatizer} is a $(1,\frac14)$-correct automatizer.
\end{definition}
\begin{remark}
For recursively enumerable $L$, conditions~\ref{def:aut:diverge} and~\ref{def:aut:comp} can be easily enforced at the cost of a slight overhead in time by running $L$'s semidecision procedure in parallel.
\end{remark}

In what follows, all automatizers are for the same problem $(D,L)$.

\begin{definition}
The \emph{time} spent by automatizer $A$ on input $(x,d)$
is defined as the median time
$$
t_A(x,d)=
\min\left\{t\in\mathbb{N}\;\bigg|\;\Pr_A\{\text{$A(x,d)$ stops in time at most $t$}\}\ge\frac12\right\}.$$
We will also use a similar notation for ``probability $p$ time'':
$$
t^{(p)}_A(x,d)=
\min\left\{t\in\mathbb{N}\;\bigg|\;\Pr_A\{\text{$A(x,d)$ stops in time at most $t$}\}\ge p\right\}.$$
\end{definition}

\begin{definition}
Automatizer $S$ simulates automatizer $W$ if
there are polynomials $p$ and $q$ such that
for every $x\in L$ and $d\in\mathbb{N}$,
$$
{t_S(x,d)}
\le \max_{d'\le q(d\cdot|x|)} p(t_W(x,d')\cdot|x|\cdot d).
$$
\end{definition}

\begin{definition}
An \emph{optimal} automatizer is one that simulates
every other automatizer.
\end{definition}

\begin{definition}
Automatizer $A$ is \emph{polynomially bounded}
if there is a polynomial $p$ such that
for every $x\in L$ and every $d\in\mathbb{N}$,
$$
t_A(x,d) \le p(d\cdot|x|).
$$
\end{definition}

The following proposition follows directly from the definitions.
\begin{proposition}\label{prop:sim+opt}\hfill
\begin{enumerate}
\item
If $W$ is polynomially bounded
and is simulated by $S$, then $S$ is polynomially bounded too.
\item
An optimal automatizer is not polynomially bounded
if and only if no automatizer is polynomially bounded.
\end{enumerate}
\end{proposition}

\section{Optimal automatizer}\label{sec:optimal}

The optimal automatizer that we construct
runs all automatizers in parallel
and stops when the first of them stops
(recall Levin's optimal algorithm
for SAT \cite{Levin-optimal}).
A major obstacle to this simple plan is
the fact that it is unclear how to
enumerate all automatizers efficiently
(put another way, how to check whether
a given algorithm is a correct automatizer).
The plan of overcoming this obstacle
(similar to constructing a complete public-key cryptosystem 
\cite{Harnik-et-al} (see also \cite{GHP-complete}))
is as follows:
\begin{itemize}
\item Prove that w.l.o.g. a correct automatizer
is very good: in particular, amplify its probability of success.
\item Devise a ``certification'' procedure that
distinguishes very good automatizers from incorrect
automatizers with overwhelming probability.
\item Run all automatizers in parallel,
try to certify automatizers that stop,
and halt when the first automatizer passes the check.
\end{itemize}

The amplification is obtained by repeating
and the use of Chernoff bounds.
\begin{proposition}[Chernoff bounds (see, e.g., {\protect\cite[Chapter 4]{MR})}]{}\hfill\newline{}Let $X_1,X_2,\ldots,X_n\in\{0,1\}$ be independent random variables.
Then if $X$ is the sum of $X_i$ and if $\mu$ is $\mathbf{E}[X]$, for any $\delta$, $0 < \delta \le 1$:
\begin{equation*}
\Pr \{ X < (1-\delta) \mu\} < e^{-\mu\delta^2/2}, \qquad
\Pr \{ X > (1+\delta) \mu\} < e^{-\mu\delta^2/3}.
\end{equation*}
\end{proposition}
\begin{corollary}
Let $X_1,X_2,\ldots,X_n\in\{0,1\}$ be independent random variables.
Then if $X$ is the sum of $X_i$ and if $1\ge \mu_1\ge \mathbf{E}[X] \ge \mu_2 
\ge 0$, for any $\delta$, $0 < \delta \le 1$:
\begin{equation*}
\Pr \{ X < (1-\delta) \mu_2\} < e^{-\mu_2\delta^2/2}, \qquad
\Pr \{ X > (1+\delta) \mu_1\} < e^{-\mu_1\delta^2/3}.
\end{equation*}

\end{corollary}

\newcommand{\amplerrdiv}{48}
\newcommand{\amplsimdiv}{64}
\begin{lemma}[amplification]\label{lem:amp}
Every automatizer $W$ is simulated by
a $(4,e^{-m/{\amplerrdiv}})$-correct automatizer $S$,
where $m\in\mathbb{N}$ may depend at most polynomially
on $d\cdot |x|$ (for input $(x,d)$).
Moreover, 
there are polynomials $p$ and $q$ such that
for every $x\in L$ and $d\in\mathbb{N}$,
\begin{equation}\label{eq:strongsim}
{t^{(1-e^{-m/{\amplsimdiv}})}_S(x,d)}
\le \max_{d'\le q(d\cdot|x|)} p(t_W(x,d')).
\end{equation}
\end{lemma}
\begin{proof}
$S(x,d)$ runs $m$ copies of $W(x,4d)$
in parallel and stops as soon as the $\frac38$ fraction
of copies stop.

By Chernoff bounds, $S$ is $(4,e^{-m/{\amplerrdiv}})$-correct.
\ourcomment{
For $x\notin L$ that is not erroneous for the old $W$,\\
$\mu=\mathbf{E}(\text{stopped})< \frac{m}4$.\\
$\Pr\{\text{stopped}\ge\frac{3m}8\}=
<e^{-m/{\amplerrdiv}}.$
}
The ``strong'' simulation condition (\ref{eq:strongsim}) is satisfied because
by Chernoff bounds the running time of the fastest $\frac{3}{8}$ fraction
of executions is less than median time with probability at least $1-e^{-m/{\amplsimdiv}}$.%
\end{proof}
\ourcomment{
For $x\in L$, let $T=$old median, $\chi_i=i$'th execution had time $\le T$.\\
$\mu=\mathbf{E}\sum_i\chi_i \ge \frac{m}2$.\\
$\Pr\{\text{new $S$ does not stop in time $\le T$}\} \le \Pr\{\sum_i\chi_i < \frac{3m}8\}\le$\\
$
< e^{-m/{\amplsimdiv}}$.
}

\begin{theorem}[optimal automatizer]\label{th:opt}
Let $(D,L)$ be a distributional proving problem,
where $L$ is recursively enumerable and
$D$ is polynomial-time samplable, i.e.,
there is a polynomial-time randomized Turing machine
that given $1^n$ on input
outputs $x$ with probability $D_n(x)$
for every $x\in\{0,1\}^n$.
Then there exists an optimal automatizer for $(D,L)$.
\end{theorem}
\begin{proof}
For algorithm $A$, we say that it is \emph{$(\lambda,\epsilon)$-correct for input length $n$
and parameter $d$} if
it it satisfies condition~\ref{def:aut:corr} of Definition~\ref{def:aut}
for $n$ and $d$.
If an algorithm is $(\lambda,\epsilon)$-correct for every $n$ (resp., every $d$),
we omit $n$ (resp., $d$).

In order to check an algorithm for correctness,
we define a \emph{certification} procedure
that takes an algorithm $A$ and distinguishes 
between the cases where $A$ is $(4,{\errlowsubst})$-correct for given $n,d$
(from Lemma~\ref{lem:amp} we know that one can assume such correctness)
or it is not $(1,{\errhighsubst})$-correct
($(1,{\errhighsubst})$-correct automatizers suffice for the correctness of further constructions).
W.l.o.g. we may assume that 
\begin{equation}
\text{$A$ satisfies conditions~\ref{def:aut:diverge} and~\ref{def:aut:comp} of Definition~\ref{def:aut}}\label{assumption}
\end{equation}
(for the latter condition, notice that $L$ is recursively
enumerable and one may run its semidecision procedure in parallel).

The certification procedure has a subroutine $\textsc{Test}$
that estimates the probability of $A$'s error simply by
repeating $A$ and couting its faults.

\bigskip
\noindent$\textsc{Test}(A,x,d',T,l,f)$:
\begin{enumerate}
\item Repeat for each $i\in\{1,\ldots,l\}$
\begin{enumerate}
\item If $A(x,d')$ stops in $T$ steps, let $c_i=1$; otherwise let $c_i=0$.
\end{enumerate}
\item If $\sum_i c_i \ge l/f$, then reject; otherwise accept.
\end{enumerate}

\begin{lemma}\label{lem:test}For every $A,x,d',T,l,f$,
\begin{enumerate}
\item
If $A(x,d')$ stops with probability less than $\testerrflowproved$,
then $\textsc{Test}$ will reject it with probability
less than 
$\testerrrej$.
\item
If $A(x,d')$ stops in time at most $T$ with probability more than $\errfhigh$,
then $\textsc{Test}$ will accept it with probability 
less than $\testerracc$.
\end{enumerate}
\end{lemma}
\begin{proof}
Follows directly from Chernoff bounds.%
\end{proof}
\ourcomment{
1. $\mu = \mathbf{E}{\sum} < l/(1.01f)$.\\
$\Pr\{\sum \ge l/f\} \le \Pr\{\sum > (1+0.01)\mu\} < e^{-10^{-4}\mu/3}$.\\
2. $\mu = \mathbf{E}{\sum} > l/(0.99f)$.\\
$\Pr\{\sum < l/f\} \le \Pr\{\sum < (1-0.01)\mu\} < e^{-10^{-4}\mu/2}$.
}

\bigskip
\noindent$\textsc{Certify}(A,n,d',T,k,l,f)$:
\begin{enumerate}
\item Repeat for each $i\in\{1,\ldots,k\}$
\begin{enumerate}
\item Generate $x_i$ according to $D_n$.
\item If $\textsc{Test}(A,x_i,d',T,l,f)$ rejects, let $b_i=1$; otherwise let $b_i=0$.
\end{enumerate}
\item If $\sum_i b_i \ge k/(2d')$, then reject; otherwise accept.
\end{enumerate}

\begin{lemma}\label{lem:certify}
Let $d,n,T\in\mathbb{N}$.
Let $A$ be an algorithm pretending to be an automatizer.
Run $$\textstyle\textsc{Certify}(A,n,d',T,k,l,f).$$
Then 
\begin{enumerate}
\item
If $A$ is $(4,{\testerrflowused})$-correct, 
then $A$
is accepted by $\textsc{Certify}$ 
almost for sure,
failing with probability less than 
${\certerrrej}$.

\item
Let $A^T$ be a restricted version of $A$
that behaves similarly to $A$ for $T$ steps
and enters an infinite loop afterwards.
If $A^T$ is not $(1,{\errfhigh})$-correct
for length $n$ and parameter $d$,
then $A$
is accepted by $\textsc{Certify}$
with probability less than 
${\certerracc}$.
\end{enumerate}
\end{lemma}
\begin{proof}
1. Let $\Delta=\{x\in\mathop{\mathrm{Im}} D_n\;|\; \Pr\{A(x,d)=1\}>{\testerrflowused} \}$.
By assumption, $D_n(\Delta)<\frac1{4d'}$.

The certification procedure takes $k$ samples from $D_n$.
For every sample $x_i\in\overline{L}\setminus\Delta$, the probability that
the corresponding $b_i$ equals 1 is less than
$\testerrrej$.
Thus, the probability that there is a sample $x_i$
from $\overline{L}\setminus\Delta$ 
that yields $b_i=1$
is less than
$k \cdot \testerrrej$.
Denote this unfortunate event by $E$.
If it does not hold, only samples from $\Delta$ may cause $b_i=1$
and by Chernoff's bound
\[
 \Pr \{ \sum_i b_i \ge k/(2d')\;|\; \overline{E} \}
 <
 e^{-\frac{k}{12d'}}.
\]
Thus, the total probability of reject is as claimed.

2. Let $\Delta=\{x\in\mathop{\mathrm{Im}} D_n\;|\; \Pr\{A(x,d)=1\}>{\errfhigh} \}$.
By assumption, $D_n(\Delta)\ge\frac1{d'}$.

The certification procedure takes $k$ samples from $D_n$.
For every sample $x_i\in\Delta$, the probability that
the corresponding $b_i$ equals 0 is less than
$\testerracc$.
Thus, the probability that there is a sample $x_i$
from $\Delta$ 
that yields $b_i=0$
is less than
$k \cdot \testerracc$. 
Denote this unfortunate event by $E$.
Assuming it does not hold only samples outside $\Delta$
may cause $b_i=0$ and by Chernoff's bound
\[
 \Pr \{ \sum_i b_i < k/(2d')\;|\; \overline{E} \}
 <
 e^{-\frac{k}{8d'}}.
\]

\vspace{-3mm}
\end{proof}

We now define the optimal automatizer $U$.
It works as follows:

\bigskip
\noindent$U(x,d)$:
\begin{enumerate}
\item Let 
\begin{eqnarray*}
n&=&|x|,\\
                           d'&=&\substd,\\ 
			   f&=&\substf,\\ 
                           k&=&\substk,\\ 
                           l&=&\substl. 
\end{eqnarray*}
\item Run the following processes for $i\in\{1,\ldots,\log_*n\}$ in parallel:
\begin{enumerate}
\item Run $A_i(x,d')$, the algorithm with Turing number $i$
      satisfying assumption~(\ref{assumption}),
      and compute the number of steps $T_i$ made by it
      before it stops.
\item If $\textsc{Certify}(A_i,n,d',T_i,k,l,f)$ accepts,\\
      then output $1$ and stop $U$ (all processes).
\end{enumerate}
\item If none of the processes has stopped, go into an infinite loop.
\end{enumerate}
\paragraph{Correctness.}
We now show that $U$ errs with probability less than 1/4.


What are the inputs 
that cause $U$ to error?
For every such input $x$ there exists $i\le\log_*n$ such that
\begin{equation}
u^i_x=\sum_{T=1}^\infty p^i_{x,T}c_T^i \ge \frac1{4\log_*n},\label{eq:bad}
\end{equation}
where
\begin{eqnarray*}
&p^i_{x,t}&=\Pr\{\text{$A_i(x,d')$ stops in exactly $t$ steps}\},\\
&c^i_t&=\Pr\{\text{$\textsc{Certify}(A_i,n,d',t,k,l,f)$ accepts}\}.
\end{eqnarray*}
Let $E_i$ be the set of inputs $x\notin L$ satisfying inequality (\ref{eq:bad}).

We claim that $D(E_i)<\frac1{d\log_*n}$,
which suffices to show the $(1,1/4)$-correctness.

Assume the contrary. 
Let $T_i=\min\{t\;|\;c^i_t<{\certerracc}\}$.
Note that
by Lemma~\ref{lem:certify}
$A_i^{T_i-1}$
is $(1,\errfhigh)$-correct for $n$ and $d'$,
i.e.,
$$\Pr_{x\gets D_n}\{\sum_{T<T^i_*} p^i_{x,T}>
{\errfhigh}
\}<
\frac1{d'}.$$
We omit $i$ and $n$ in the estimations that follow.
Here is how we get a contradiction:
\begin{multline*}
\frac1{4d\log_*^2n}\le 
\frac{D(E_i)}{4\log_*n}= 
\sum_{x\in E_i} \frac1{4\log_*n} D(x) \le 
\sum_{x\in E_i} u_x D(x)\le\\ 
\sum_{x\notin L} u_x D(x) =
\sum_{x\notin L} \sum_{T=1}^\infty p_{x,T} c_T D(x) =\\
\sum_{x\notin L} \left(
 \sum_{T < T_*} p_{x,T} c_T D(x) +
 \sum_{T\ge T_*} p_{x,T} c_T D(x) 
\right)\le\\ 
\sum_{T<T_*} \left(
 \sum_{x\notin L,\ \sum\limits_{t<T_*} p_{x,t}\le {\errfhigh}} p_{x,T}  D(x) +
 \sum_{x\notin L,\ \sum\limits_{t<T_*} p_{x,t} >  {\errfhigh}} p_{x,T}  D(x)
\right)\\
\mbox{\hfill}+\certerracc\le\\
{\errfhigh}
+\frac1{d'}+\certerracc
<
\errhighsubst+
\frac1{\substd}+
\certerraccsubst
=
\frac1{4d\log^2_*n}.
\end{multline*}

\paragraph{Simulation.}
Assume we are give a correct automatizer $A^s$.
Plug in $m=\amplerrdiv\cdot\ln(\inverrlowsubst)$ into Lemma~\ref{lem:amp}.
The lemma yields that $A^s$ is ``strongly'' simulated by a
$(4,{\errlowsubst})$-correct automatizer $A$.
It remains to estimate,
for given ``theorem'' $x\in L$,
the (median) running time of $U$
in terms of $t^{(1-e^{-m/\amplsimdiv})}_A(x,d)=
t^{(1-\frac1{(\inverrlowsubst)^{3/4}})}_A(x,d)$
(as we know that the latter is bounded
by $\max\limits_{d'\le q(d\cdot|x|)} p(t_{A^s}(x,d'))$ for a polynomials $p$ and $q$).

Since the definition of simulation is asymptotic,
we consider only $x$ of length greater than the Turing number
of $A$. 
By Lemma~\ref{lem:certify},
$A$ is not certified with probability less than
$\certerrrej\le\certerrrejsubst$.
If $A$ is certified, $U$ stops in time
upper bounded by a polynomial of the time spent by $A$
with an overhead polynomial in $|x|$ and $d$
for running other algorithms and the certification procedures.
Thus the median time $t_U(x,d)$
is bounded by a polynomial in $|x|$, $d$,
and $t_A^{(\frac12+\certerrrejsubst)}(x,d)
\le
t^{(1-\frac1{(\inverrlowsubst)^{3/4}})}_A(x,d)$.%
\end{proof}

\section{Heuristic proof systems}\label{sec:systems}

In this section we define proof systems that make errors
(claim a small fraction of wrong theorems).
We consider automatizable systems of this kind and show that every such system
defines an automatizer taking time at most polynomially larger
than the length of the shortest proof in the initial system.
This shows that automatizers form a more general notion
than automatizable heuristic proof systems.
The opposite direction is left as an open question.

\begin{definition}
Randomized Turing machine
$\Pi$ is a \emph{heuristic proof system} for distributional proving
problem $(D,L)$ if it satisfies the following conditions.
\begin{enumerate}
\item The running time of $\Pi(x,w, d)$ is bounded by a
polynomial in $d$, $|x|$, and $|w|$.
\item (Completeness) For every $x\in L$ and every $d\in\mathbb{N}$, there exists a
string $w$ such that $\Pr\{\Pi(x,w,d)=1\}\ge \frac12$. Every such string $w$
is called a $\Pi^{(d)}$-proof of $x$.
\item (Soundness) 
$\Pr_{x\gets D_n} \{\exists w: \Pr\{\Pi(x,w,d)=1\}> \frac14\}<\frac{1}{d}$.
\end{enumerate}
\end{definition}

\begin{definition}
Heuristic proof system is \emph{automatizable} if there is
a randomized Turing machine $A$
satisfying the following conditions.
\begin{enumerate}
\item For every $x\in L$ and every $d\in\mathbb{N}$, 
with probability at least $\frac12$ algorithm $A(x,d)$ outputs
a correct $\Pi^{(d)}$-proof of size
bounded by 
a polynomial in $d$, $|x|$, and $|w|$,  where $w$
is the shortest $\Pi^{(d)}$-proof of $x$.
\item The running time of $A(x,d)$ is bounded by 
a polynomial in $|x|$, $d$, and the size of its own output.
\end{enumerate}
\end{definition}

\begin{definition}
We say that heuristic proof system $\Pi_1$ \emph{simulates} heuristic proof
system $\Pi_2$ if there exist polynomials $p$ and $q$ such that for every 
$x\in L$, the shortest
$\Pi_1^{(d)}$-proof of $x$ has size at most 
\[
p(d\cdot|x|\cdot \max_{d'\le q(|x|d)}\{\mbox{the size of the shortest $\Pi_2^{(d')}$-proof of $x$}\}).
\]
\end{definition}
Note that this definition essentially ignores proof systems
that have much shorter proofs for some inputs than the inputs themselves.
We state it this way for its similarity to the automatizers case.

\begin{definition}
Heuristic proof system $\Pi$ is \emph{polynomially bounded}
if there exists a polynomial
$p$ such that for every $x\in L$ and every $d\in\mathbb{N}$,
the size of the shortest 
$\Pi^{(d)}$-proof of $x$ is bounded by $p(|x|d)$.
\end{definition}

\begin{proposition}
If heuristic proof system $\Pi_1$ simulates system $\Pi_2$ and $\Pi_2$ is polynomially
bounded, then $\Pi_1$ is also polynomially bounded.
\end{proposition}

We now show how automatizers and automatizable heuristic proof systems are related.

Consider automatizable proof system $(\Pi, A)$ for
distributional proving problem $(D,L)$
with recursively enumerable language $L$. Let us consider the following 
algorithm $A_{\Pi}(x,d)$:
\begin{enumerate}
\item Execute 1000 copies of $A(x,d)$ in parallel.\\
For each copy,
\begin{enumerate}
\item if it stops with result $w$, then 
\begin{itemize}
\item
execute $\Pi(x,w,d)$ 10000 times;
\item
if there were at least 4000 accepts of $\Pi$ (out of 10000),
stop all parallel processes and output $1$.
\end{itemize}
\end{enumerate}

\item Execute the enumeration algorithm for $L$;
output 1 if this algorithm says that $x\in L$;
go into an infinite loop otherwise.
\end{enumerate}

\begin{proposition}
If $(\Pi, A)$ is a (correct) heuristic automatizable proof system
for recursively enumerable language $L$,
then $A_{\Pi}$ is a correct automatizer for $x\in L$ and 
$t_{A_{\Pi}}(x,d)$
is bounded by polynomial in size of the shortest $\Pi_d$-proof of $x$.
\end{proposition}
\begin{proof}
\emph{Soundness (condition~\ref{def:aut:corr} in Def.~\ref{def:aut}).}
Let $\Delta_n=\{x\in \overline{L}\mid \exists w: \Pr\{\Pi(x,w,d)=1\}> \frac14\}$.
By definition, $D_n(\Delta_n)<\frac{1}{d}$.
For $x\in \{0,1\}^n \setminus \Delta_n$ and specific $w$,
Chernoff bounds imply that 
$\Pi(x,w,d)$ accepts in $0.4$ or more fraction of executions with
exponentially small probability, which remains much smaller
than $\frac14$ even after
multiplying by 1000.

\emph{Completeness (conditions~\ref{def:aut:comp} and~\ref{def:aut:diverge} in Def.~\ref{def:aut})} is guaranteed by the execution of the semi-decision procedure for $L$.

\emph{Simulation.}
For $x\in L$,
the probability that $A$ errs 1000 times is negligible (at most $2^{-1000}$).
Thus
with high probability at least one of 
the parallel executions of $A(x,d)$ outputs a correct $\Pi_d$-proof of
size bounded by a polynomial in the size of the shortest $\Pi_d$-proof of $x$. 
For $x\in L$ and (correct) $\Pi^{(d)}$-proof $w$,
Chernoff bounds imply that 
$\Pi(x,w,d)$ accepts in at least $0.4$  fraction of executions with
probability close to $1$.
Therefore, $t_{A_{\Pi}}(x,d)$
is bounded by  a polynomial in $|x|$, $d$,
and the size of the shortest $\Pi_d$-proof of $x$.%
\end{proof}

\section{Further research}

One possible direction is to show that automatizers are equivalent
to automatizable heuristic proof systems or, at least,
that there is an optimal automatizable heuristic proof system.
That may require some tweak in the definitions, because
the first obstacle to proving the latter fact is
the inability to check a candidate proof system for
the non-existence of a much shorter (correct) proof
than those output by a candidate automatizer.

Also Kraj{\'{\i}}{\v{c}}ek and Pudl{\'{a}}k \cite{KP}
and Messner \cite{Mes}
list equivalent conditions for the existence of (deterministic)
optimal and $p$-optimal proof systems. It seems promising
(and, in some places, challenging) to prove similar statements
in the heuristic setting.

\section*{Acknowledgements}

During the work on the subject, we discussed it with many people.
Our particular thanks go to (in the alphabetical order)
Dima Antipov, Dima Grigoriev, and Sasha Smal.

\small

\bibliographystyle{alpha}
\bibliography{autpps}

\newpage
\strut

\end{document}